\documentclass[conference,letterpaper]{IEEEtran}
\IEEEoverridecommandlockouts
\usepackage{cite}
\usepackage{amsmath,amssymb,amsfonts}

\usepackage{graphicx}
\usepackage{textcomp}
\usepackage{xcolor}
\usepackage{fancyhdr}

\usepackage{multirow}
\usepackage{diagbox} 
\usepackage{booktabs} 

\usepackage{algpseudocode}
\usepackage{algorithm}
\usepackage{amsthm}

\newtheorem{theorem}{Theorem}

\newtheorem{lemma}{Lemma}

\newcommand{\etal}{\textit{et al}. }

\usepackage{enumitem}

\def\BibTeX{{\rm B\kern-.05em{\sc i\kern-.025em b}\kern-.08em
    T\kern-.1667em\lower.7ex\hbox{E}\kern-.125emX}}
\begin{document}

\title{Eunomia: A Permissionless Parallel Chain Protocol Based on Logical Clock}

\author{\IEEEauthorblockN{Jianyu Niu}
\IEEEauthorblockA{School of Engineering, University of British Columbia (Okanagan Campus), Kelowna, Canada\\
jianyu.niu@ubc.ca}
}

\maketitle
\cfoot{}
\pagestyle{fancy}
\rfoot{\thepage}

\begin{abstract}
The emerging parallel chain protocols represent a breakthrough to address the scalability of blockchain. By composing multiple parallel chain instances, the whole systems' throughput can approach the network capacity. How to coordinate different chains' blocks and to construct them in a global ordering is critical to the performance of parallel chain protocol. However, the existing solutions use either the global synchronization clock with the single-chain bottleneck or pre-defined ordering sequences with distortion of blocks' causality to order blocks. Besides, the prior ordering methods rely on that honest participants faithfully follow the ordering protocol but remain silent for any denial of ordering (DoR) attack. 

On the other hand, the conflicting transactions included in the global block sequence make Simple Payment Verification (SPV) difficult to tell confirmed transactions. Clients usually need to store a full record of transactions to distinguish the conflictions and tell whether transactions are confirmed. However, the requirement for a full record greatly hinders blockchains' application, especially for mobile scenarios. 

In this technical report, we propose Eunomia, which leverages logical clock and fine-grained UTXO sharding to realize a simple, efficient, secure, and permissionless parallel chain protocol. 
By observing the characteristics of the parallel chain, we find the blocks ordering issue in the parallel chain has many similarities with the event ordering in the distributed system. Eunomia thus adopts the ``virtual'' logical clock, which is optimized to have the minimum protocol overhead and runs in a distributed way. In addition, Eunomia combines the mining incentive with block ordering, providing incentive compatibility against the DoR attack. What's more, the fine-grained UTXO sharding does well solve the conflicting transactions in the parallel chain and is shown to be SPV-friendly. 
\end{abstract}


\section{Introduction} \label{sec:intro}
\noindent In 2008, Nakamoto invented the seminal \emph{blockchain} protocol which uses the \emph{Nakamoto Consensus (NC)} to realize a public, immutable and distributed ledger, Bitcoin\cite{nakamoto2012bitcoin}. In NC, participates, called miners are allowed to generate a new block --- a collection of transactions by solving computational puzzles, known as \emph{Proof-of-Work (PoW)}, and can reach an agreement of a sequence of blocks by following the \emph{longest chain rule} (\emph{LCR}). Although remarkably simple, NC enjoys very brilliant properties: it does not require participants' identities as a setup --- a fully permissionless setting and can work as long as more than half of the computation power controlled by honest participates. Due to this, more than six hundred digital currencies leverage NC or its components, \emph{e.g.} PoW, to maintain consensus \cite{mapofcoins, Zhang2019CommonMetrics}. 

Unfortunately, NC and its variants suffer from low transaction throughput (\emph{e.g.}, $7$ TPS\footnote{TPS is short for transaction per second.} in Bitcoin and $20$ TPS in Ethereum), resulting from the inherently speed-security trade-off (\emph{i.e.}, block is generated sequentially and the generation rate needs to be relatively small for security constraints\cite{garay2015, Pass2017, garay2017, Kiffer2018, Zhang2019CommonMetrics}). Despite the improvements made by \cite{Zhang2017, bitcoinng, Sompolinsky2016SPECTREAF, Sompolinsky2015ghost}, these NC-based protocols are still designed based on this paradigm and cannot break the security limits to achieve the optimal throughput, \emph{i.e.}, the maximum supported by the underlying communication network. 

\emph{Parallel chain} represents an exciting breakthrough, which first steps out the design paradigm and pushes blockchains' throughput to the optimal. Surprisingly, the parallel chain achieves this with a straightforward idea --- the whole system's throughput can promote by increasing the number of chains. In doing so, there are three main challenges. First, each chain's security will not be influenced as the number of chains increases. This has been well-solved by the $m$-for-$1$ PoW \cite{garay2015} (introduced in Sec.~\ref{subsec:parallelmining}). 

Second, blockchain is essentially a distributed, append-only ledger that outputs a globally ordered transaction sequence.
For example, in NC the main chain, \emph{i.e.}, the longest path in a blocktree structure  \cite{nakamoto2012bitcoin} intrinsically provides a global order of all included blocks, resulting in a sequence of transactions. However, without considering cross-chain hash references (See Sec.~\ref{subsec:parallelmining}) in the parallel chain, participants observe multiple blocktrees, forming a blockforest. Moreover, there does not exist any explicit ordering relationships between mined blocks of different blocktrees because in $m$-for-$1$ PoW blocks randomly extends one of the parallel chains. Thus, how to globally order blocks in a sequence is the key to design a secure and high-performance parallel chain system. 

To meet the challenge, a few elegant designs have been proposed in the recent works \cite{fitzi2018parallel, bagaria2018deconstructing, yu2018ohie}, among which \cite{fitzi2018parallel, bagaria2018deconstructing} use one unique chain, i.e., ``synchronization chain'' to either be referenced as a globally synchronized clock or reference other chains' blocks in sequence, providing the ordering relationship. However, once the single unique chain grows slowly, the whole system's block ordering is delayed, and transaction latency increases. 
In \cite{yu2018ohie}, Yu \etal adopts a simple pre-defined sequence to order blocks and propose weighted blocks to balance chains' length. Specifically, a new attachment block containing fields $(rank, next\_rank)$ is adopted. However, using the pre-defined sequence may ruin the causality of blocks. In other words, a block may be inserted after a future block mined on it into the global block sequences, which provides the possibility for new attack vectors (e.g., a later mined block including the same transactions for winning transaction fees). What's more, the additional attachment blocks increase the protocol's overhead and make the analysis complicated. Their limitations make us believe the ordering issue of the parallel chain has not been fully explored.

Third, in the parallel chain, the blocks concurrently mined by multiple participants in different chains may contain the same transactions more than once, especially those transactions with high fees \cite{nakamoto2012bitcoin}. It is easy to see that the resulted transaction redundancy decreases the whole systems' effective throughput. More importantly, the Byzantine clients can issue multiple conflicting transactions, which contain the same UTXO as input and have different output, into the network. This problem is also known as a double-spending attack in NC. Due to the concurrency of blocks, conflicting transactions may be included in the global ordering block sequence. Thus, the first transaction is confirmed, whereas the rest of conflicting transactions are invalidated. In other words, clients need to have a full record of transactions to tell confirmed transactions. Note that as blocks are created sequentially in NC, only one of the conflicting transactions can be eventually included in the main chain. The differences lead to that light client, i.e., client just store blocks headers, cannot use SPV to verify transactions anymore. Although the prior works, e.g., transaction sharding in \cite{fitzi2018parallel} and colored transaction scheduling in \cite{bagaria2018deconstructing} can solve redundant transactions, there is no solution for conflicting transactions or feasible design for light clients and SPV. All of those motivate our work.



\subsection{Objective and Contributions}
\noindent In this technical report, we design a permissionless parallel chain protocol, namely Eunomia, which addresses the total block ordering and conflicting transaction issues. We try to make Eunomia's design modular and leverage the simple and thoroughly analyzed NC \cite{garay2015, Pass2017, garay2017, Kiffer2018, Zhang2019CommonMetrics} to build each chain instance. It is easy to replace some components, \emph{e.g.} using PoS or Ghost \cite{Sompolinsky2015ghost}, to design a new parallel chain protocol. Eunomia is proved to tolerate up to $1/2$ adversarial computation power and meanwhile, to provide the optimal throughput. Our contributions in this technical report are twofold. 

\begin{enumerate}[label=(\arabic*)]
\item To the best of knowledge, we are the first to use ``logical clock'' in distributed system \cite{Lamport1978, fidge1987timestamps, mattern1988virtual} to solve the global block ordering issues in the parallel chain. Our adoption is based on the following two observations: $i$) each chain can be an analogy to one process in a machine, and the included blocks can be transformed into a metaphor of events \cite{Lamport1978}; $ii$) blocks' contained hash references, \emph{i.e.}, the logical link constructing blocks into chains, can provide evidence for blocks' causality and synchronize chains' logical clocks. Based on the observations, 
we propose the ``virtual logical clock'', which can be used to timestamp blocks and be computed according to the directed graph of connected blocks. 
In our approach, all chains are equal, and chains' synchronization does not rely on any unique chain in a distributed way. In addition, our ordering protocol is shown to be as simple as \cite{yu2018ohie}, but order blocks according to their causality without using any additional block. What's more, our protocol combines mining incentive with blocks' ordering, and so any denial of ordering attack will bring economic loss to the attacker. 

\item In a UTXO-based blockchain system, the objects involved in the record can be divided into three levels: UTXO, transaction, and client (listed from bottom to up). A client can issue multiple transactions, and a transaction can contain multiple unspent UTXO as input and create new UTXOs. In this technical report, we propose a new fine-grained UTXO sharding, which is different from the coarse sharding based on transactions and clients. The UTXO sharding allows us to trace and to control clients' UXTO in a more detailed way, and so do well in solve the conflicting transaction issue in the parallel chain. The new sharding design is shown to be SPV-friendly, which can make the light client as secure as it in NC, and meanwhile do not sacrifice too many performances (e.g., transaction latency and protocol overhead). 
\end{enumerate}

\emph{Report Structure.} Sec.~\ref{sec:back} introduces the background, and Sec.~\ref{sec:model} introduces system model and design goals. Sec.~\ref{sec:sharding} presents the UTXO sharding design, and Sec.~\ref{subsec:order} gives the global ordering protocol. The security analysis is given in Sec.~\ref{sec:analysis}. 
Sec.~\ref{sec:related} provides the related work. 
Finally, Sec.~\ref{sec:conclusion} concludes the report.

\section{Background} \label{sec:back}
\noindent In this section, we first provide the necessary background of Bitcoin, then present the extended $m$-for-$1$ PoW and finally introduce logical clock.

\subsection{A Primer on Bitcoin}
\noindent In Bitcoin, there are two types of nodes: clients and miners. Clients can create and send transactions to each other, and miners can collect the received transactions, verify them, pack them in blocks. Each block in Bitcoin contains two components, a block header and a set of transactions. The block header includes a hash value of the previous block, a timestamp, a Merkle tree root of transactions and a nonce (whose role will be explained shortly). Blocks are linked together by the hash references and form a chain structure. The chain structure decides a unique sequence of blocks and further provides a sequence of transactions. 

\noindent \textbf{Nakamoto Consensus.} 
Bitcoin relies on NC to reach an agreement of blocks. In NC, all miners obey PoW to generate blocks, in which miners need to find a value of the nonce such that the hash value of the new block has required leading zeros \cite{nakamoto2012bitcoin}. This puzzle-solving process is often referred to as \emph{mining}. Intuitively, the leading zeros determine the chance of finding a new block in each try. By adjusting the number of leading zeros, the blockchain system can control the block generation rate and then maintain a stable chain growth, e.g., one block per 10 minutes in Bitcoin.

Once a new block is produced, it will be broadcast to the entire network. In the ideal case, a block will arrive all the clients before the next block is produced. If this happens to every block, then each client in the system will have the same chain of blocks. In reality, due to the network delay, a miner may produce a new block \emph{before} he or she receives the previous block, a fork will occur where two ``child'' blocks share a common ``parent'' block.  In general, each client in the system observes a tree of blocks due to the forking. As a result, each client has to choose the longest path from the tree as the main chain, known as the longest chain rule. The common prefix of all the main chains is called the \emph{system main chain}---a key concept that will be used in our analysis. 

\noindent \textbf{UTXO Model.}
The Unspent Transaction Output (UTXO) model is widely used in Bitcoin and other blockchains systems \cite{omniledger,Luu2016}. In this model, a transaction consumes unspent outputs as inputs from previous valid transactions and creates new unspent outputs that can be used in a future transaction. Specifically, a transaction allows senders to use multiple UTXOs as inputs, and at most two outputs: one for the payment, and one returning the change, if any, back to the sender \cite{nakamoto2012bitcoin}. It is easy to see that a client owns multiple UTXOs at the same time, and his or her account can be worked out by summing up all the belonging UTXOs, which is different from account/balance model\cite{buterin2014next}. 

In this technical report, we leverage the LCR and $m$-for-$1$ PoW (i.e., an extension of PoW introduced later) to build each chain instance in Eunomia. This is because NC is simple and has many promising properties (See Sec.~\ref{sec:intro}).  More importantly, NC has been thoroughly analyzed in various prior works \cite{garay2015, Pass2017, garay2017, Kiffer2018, Zhang2019CommonMetrics} and proved to provide the end-to-end security. For the transaction model, our system uses the UTXO model for its simplicity and parallelizability. A client can simultaneously issue multiple transactions with different unspent UTXOs. 

\subsection{The $m$-for-$1$ PoW} \label{subsec:parallelmining}
By extending PoW, a miner can simultaneously generate blocks on $m$ chains, and each time a block extends one chain. This extended PoW is known as $m$-for-$1$ PoW \cite{garay2015,fitzi2018parallel}. See Fig.~\ref{fig:mining} for an illustration. In $m$-for-$1$ PoW, each block contains $m$ hash references, each of which links to the last block in $m$ main chains. Each main chain is independently chosen by certain rules (\emph{e.g.}, LCR in NC and the heaviest subtree rule in the GHOST protocol \cite{Sompolinsky2015ghost}). When a valid nonce is found, the last $\log_2{m}$ bits of the block's hash value can index to one of the $m$ chains that the block belongs to \cite{yu2018ohie}. Here, the hash reference of the block in the belonging chain is called chain reference, and the rest $m-1$ are called cross-chain references. As the hash function can be assumed as a random oracle \cite{garay2015,garay2017}, each chain owns the same chance of receiving a valid block and then maintains the same block generation rate on average. Note that there exist many methods of generating a uniform distribution to randomly assign blocks to $m$ chains through the random oracle\cite{bagaria2018deconstructing, fitzi2018parallel}. 
The resulted same block generation rate ensures that each chain owns the same security guarantee. 
Finally, by adjusting NC's mining difficulty to $m$ times harder, each chain instance in the parallel chain can be proved to have the same security guarantee as to the single chain in NC.

\begin{figure}[!ht]
\setlength{\abovecaptionskip}{2pt}
\setlength{\belowcaptionskip}{5pt}
\centering
\includegraphics[width=0.5\linewidth]{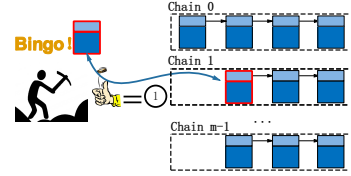}
\caption{The $m$-for-$1$ PoW for parallel chain}
\vspace{-2mm}
\label{fig:mining}
\end{figure}

\subsection{Logical Clock} \label{subsec:clock}
In the distributed system, multiple processes running on different entities need to achieve an agreement with the ordering of events, e.g., the sequences of two messages from two different processes, or to avoid simultaneously access the shared resources, e.g., printer \cite{TanenbaumSteen07}. To realize this, 
in 1978, Leslie Lamport first pointed out that the distributed system can use logical clock instead of the physical clock to capture the causal relationship between events \cite{Lamport1978}. Specifically, the logical clock is an incrementing software counter maintained in each process and synchronized by message exchanges. In \cite{Lamport1978}, he also proposed one realization, called Lamport timestamps, which can provide a partial ordering of events. Furthermore, to realize synchronization of timestamps, Lamport defined a relation called ``happens-before'': $a \rightarrow b$ means that event $a$ happens before $b$. Using Lamport's timestamps together with a breaking tie policy of process ID, all processes can achieve a total ordering of events. 
Based on the seminal work, the vector clock is proposed to capture causality between events by Fidge \cite{fidge1987timestamps} and Mattern \cite{mattern1988virtual} in 1988. 

Our system uses logical clock instead of vector clock due to its simplicity since the logical clock is easier to be updated and compared. What's more, the vector clock cannot exploit its advantage, i.e., providing correct causality between blocks, in the parallel chain because of the forked branches. 

\section{The System Design and Goals} \label{sec:model}
\noindent \textbf{System Model.} Our system model and assumptions follow the well-built formalization of blockchain protocols in \cite{garay2015,Pass2017,Kiffer2018,Pass201701}. A blockchain protocol $\Pi$ refers to an algorithm for a set of Turing Machines, also called participants, to interact with each other. The execution of $\Pi$ is directed by an environment $Z(1^{\kappa})$ (where $\kappa$ is a security parameter), which actives a set of $n$ participants as either being honest for following the protocol, or being corrupt for behaving arbitrarily. 
Particularly, corrupting participants are assumed to have at most $\rho$ fraction of $n$ participants and are controlled by an adversary $\mathcal{A}$. 

The execution of protocol $\Pi$ proceeds in rounds; at each round each participant receives messages from $Z$ (\emph{e.g.}, outstanding transactions and other participants' blocks) and makes a query to the random oracle $H$ (\emph{i.e.}, the hash computation for PoW). For each query, participants have a probability of $mp$ to create a new block, where $m$ is the number of parallel chain instances in Eunomia.  Here $mp$ is referred to as the \emph{mining hardness parameter} (controlled by the leading zeros in PoW). Particularly, the adversary $\mathcal{A}$ can have up to $\rho n$ sequential quires in each round. The NC and Eunomia protocol are referred to as $\Pi_{nak}$ and $\Pi_{euo}$, respectively. The execution of protocol can be modeled as a random variable $\text{EXEC}^{\Pi}(A,Z,\kappa)$, denoting the joint view of all participates (\emph{i.e.}, all their inputs, random coins, and messages received, including those from the random oracle). 



\noindent \textbf{Network Model.}  we assume that honest participants can broadcast messages to each other. $\mathcal{A}$ is responsible for delivering messages (\emph{e.g.}, transactions and blocks) sent by honest participants to all other participants. $\mathcal{A}$ can not modify the content of messages broadcast by honest participants, but it may delay or reorder the delivery of a message as long as it eventually delivers all messages sent by honest participants within some bound $\Delta$. Note that in reality, participants are communicating messages through a gossip network, and thus, honest participants are assumed to sufficiently connect with each other to guarantee the $\Delta$ bounded delay. What's more, to capture the impact of network capacity, we assume the network can transmit at most $\mu$ transactions per round on average. Each block size is $D$ transactions. 


\noindent \textbf{System Goal.} 
As a blockchain consensus protocol, Eunomia can provide the \emph{safety} and \emph{liveness} guarantees with less than $1/2$ adversaries, and meanwhile have the \emph{optimal} throughput performance. Particularly, in context of blockchain, \emph{safety} corresponds to \emph{consistency} \cite{yu2018ohie,garay2015,garay2017}. Thus, Eunomia aims to achieve the followings: 

\begin{enumerate}[label=(\arabic*)]
\item \textbf{Consistency.} For any pair of honest participants $P_1$, $P_2$ outputting globally confirmed block sequence $L_1$, $L_2$ at time $t_1 \leq t_2$, the protocol holds that $L_1$ is a prefix of $L_2$.

\item \textbf{Liveness.} All transactions issued by honest clients will eventually get incorporated into an ordered confirmed block in any honest participant's blockchain.

\item \textbf{Near-optimal Network Utilization.} Eunomia aims to achieve a throughput, in terms of transactions processed per second, that approaches the network capacity.
\end{enumerate}


\section{UTXO sharding and SPV} \label{sec:sharding}
\noindent Eunomia composes $m$ parallel chains. Each chain $\mathcal{C}$ is identified with a chain index $i$ ($i \in \{0,1,...,m-1\}$) and initialized with a genesis block $\mathcal{G}_i$. The miners adopt LCR and $m$-for-$1$ PoW to currently mine on $m$ chains. Each time a mined block $\mathcal{B}$ extends one chain. The index of belonging chain is denoted by $ch(\mathcal{B})$, and the referred block in chain  $C_{ch(\mathcal{B})}$ is called parent block. Each block consists of two parts: a block header and a transaction metadata. As shown in Fig.~\ref{fig:blockheader}, the block header includes a timestamp, a nonce, a hash value of synchronized block, a Merkle tree root of transaction metadata, etc. The reference of synchronized block and transaction metadata do not exist in Bitcoin and will be introduced shortly. 

\subsection{UTXO Sharding} \label{subsec:utxo}
\noindent In Eunomia, each UTXO contains a $\left \lceil \log_2(m) \right \rceil$ bits filed called sharding index. The UTXO's sharding index $i$ ($i \in [ 0, 1, ..., m-1] $) represents that the UTXO is only allowed to be spent in chain $C_i$. In other words, a transaction needs to contain the UTXOs with the same sharding index as input, whereas can generate output UTXOs with any sharding index ranging from $0$ to $m-1$. Thus, according to input UTXOs' sharding index $i$, outstanding transactions can be put into corresponding transaction sets $S_i \left ( i \in [0,1,..,m-1] \right )$. As a result, the transaction $tx$ in $\mathcal{S}_i$ will be mined in chain $C_i$. 
It is easy to see that the same transactions or the conflicting transactions cannot be included in several blocks located in different chains, leading to no transaction redundancy or transaction confliction. 

\subsection{Transactions Metadata} \label{subsec:trans}
\noindent In Eunomia, participants do not know which chain the future mined block will extend because of $m$-for-$1$ PoW. Thus, they need to prepare outstanding transactions and the last blocks' hash value $hash_i$ for each chain $\mathcal{C}_i$. When the number of chain instances $m$ grows larger (e.g., $m = 1000$), the $m$ hash references and the $m$ Merkle roots of transactions directly included in the block header will become an overwhelming overhead. 

To solve this, Eunomia leverages the widely-used Merkle Tree \cite{Merkle, merkle1, nakamoto2012bitcoin} to compress these protocol overhead. Specifically, participants first construct transactions for the same chain as a Merkle tree and compute the Merkle tree root $\mathcal{\gamma}_i$, which is similar to Bitcoin\cite{nakamoto2012bitcoin}. Then, they use the the pair $(hash_0, \mathcal{\gamma}_0)$ through $(hash_{m-1}, \mathcal{\gamma}_{m-1})$ as the Merkle tree leaves and compute the related root $\mathcal{\gamma}$. Finally, the Merkle tree root $\mathcal{\gamma}$ is contained into the block header as input to the PoW puzzle. See Fig.~\ref{fig:blockheader} for an illustration. After a block $\mathcal{B}$ is mined, the last $\log_2(m)$ bits of the block's hash value will determine its belonging chain index ${ch(\mathcal{B})}$. When broadcasting the block, the chosen transaction list, hash references for chain $\mathcal{C}_{ch(\mathcal{B})}$ and the Merkle tree proof of $(hash_{ch(\mathcal{B})}, \mathcal{\gamma}_{ch(\mathcal{B})})$ will be attached. When receiving the block, other nodes in the network can verify the included transactions and hash reference for chain $C_{ch(\mathcal{B})}$. 

\begin{figure}[!ht]
\setlength{\abovecaptionskip}{2pt}
\setlength{\belowcaptionskip}{5pt}
\centering
\includegraphics[width=0.9\linewidth]{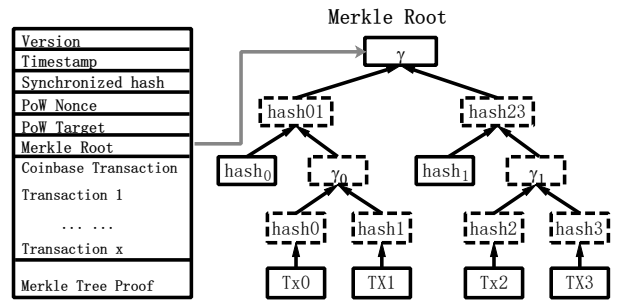}
\caption{Data structure of the block and the Merkle tree of transaction metadata.}
\vspace{-2mm}
\label{fig:blockheader}
\end{figure}

\subsection{Cross-Chain UTXO}
Recall Sec.~\ref{subsec:utxo}, clients are allowed to issue transactions containing the output UTXOs, which have different sharding indexes with the input UTXOs. In other words, these output UXTOs, namely cross-chain UTXOs, will be spent in a future transaction located in another chain. By utilizing the cross-chain UTXOs, one one hand clients can collect their owned unspent UTXOs in one chain and combine them for future transactions with a large value input, avoiding issuing multiple splitting transactions in different chains; on the other, clients can freely choose the chain with less outstanding transactions, which can balance chains' transaction workload. 

\noindent \textbf{Stale Blocks.}
For each chain instance in Eunomia, miners have the chance of mining blocks on multiple forked branches, which share the common ancestor blocks. Thus, miners in the system observe $m$ trees of blocks due to the forking. As a result, miners choose a main chain from each tree to mine according to LCR and eventually reach an agreement of $m$ system main chains. The blocks included in the system main chain are called regular blocks, and the rest are called stale blocks. In Eunomia, only the stale blocks, which provide synchronization references for regular blocks will be kept, and the rest will be discarded (See details in Sec.~\ref{subsec:propa}.). 

However, the stale block may affect the validity of cross-chain UTXOs and the subsequent transactions containing these UTXOs. For example in Figure.~\ref{fig:utxos}, if a block $A$ becomes stale and the containing transaction $tx1$ is invalidated, and then a subsequent transaction $tx2$ in block $B$ of another chain, which contains the cross-chain UTXO generated in transaction $tx1$ will be affected. In that case, block $B$ will not be invalidated, and just the transaction $tx2$ will be invalidated. Furthermore, all subsequent transactions $tx3$ and $tx4$ that relay on the cross-UTXOs in transaction $tx1$ will become invalidated.

\begin{figure}[!ht]
\setlength{\abovecaptionskip}{2pt}
\setlength{\belowcaptionskip}{5pt}
\centering
\includegraphics[width=0.9\linewidth]{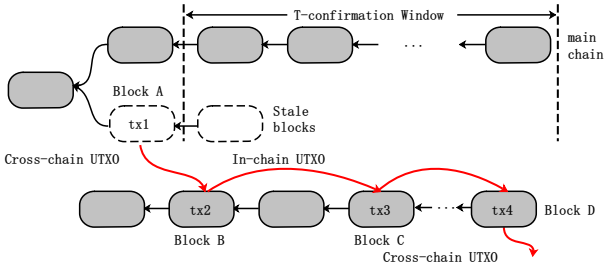}
\caption{The infectious window for Cross-chain UTXO due to orphan blocks.}
\vspace{-2mm}
\label{fig:utxos}
\end{figure}

\noindent \textbf{$T$-confirmation Window.} 
In Eunomia, a block in any chain instance is confirmed with high probability when there are $T$ successive blocks appended. It is also known as $T$-confirmation (e.g., $T = 6$ in Bitcoin). See Sec.~\ref{sec:analysis} for the proved security properties. In other words, after waiting for a $T$-block growth of the chain, miners can tell whether a block is confirmed or stale with high probability, which we refer it to as $T$-confirmation window. Thus, to reduce the domino effects of invalidated cross-chain UTXOs, one simple solution is that miners do not include any transactions, which use the cross-chain UTXOs as input, in their blocks until the originate block is confirmed. However, this method will double the latency of the transaction containing input cross-chain UTXOs because of the additional $T$-confirmation for originate transactions. 

To solve this, we allow the subsequent transactions interleaving with the cross-chain UTXOs to be included in the chain first, and the corresponding confirmation is settled later. We call such a confirmation, \textbf{Post-Confirmation}. It is easy to see post-confirmation can reduce the transactions' latency. Another reason for adopting post-confirmation is based on the observation that the domino effect of invalidated cross-chain UTXO cannot reach farther with the $T$-confirmation properties. For example, in Figure.~\ref{fig:utxos}, in the $T$-confirmation window, as miners do not know the transaction $tx1$ and its output cross-chain UTXOs will be invalidated in the future, they can accept the block, which contains the transaction interleaving with the originate cross-chain UTXOs. Once the originate chain grows beyond the $T$-confirmation window, miners will know the cross-chain UTXOs and any transaction interleaving with the cross-chain UTXOs are invalid, and then miners will refuse to accept any blocks containing these invalid transactions. In other words, an invalidated cross-chain UTXOs will only affect the involving transactions in the $T$-confirmation window. 


\subsection{Simple Payment Verification (SPV)} \label{subsec:spv}
In Eunomia, there are two kinds of clients: the full client with a full record of transactions and the light client with only block headers, hash references of parent block, and Merkle tree proof for transaction metadata. The hash references of parent block and Merkle tree proof for transaction metadata provide evidence for blocks to be chained. It is easy for a full client to check the validity of transactions by backtracing all input UTXOs in the globally confirmed transaction sequences (which will be introduced shortly). By contrast, because of the invalidated cross-chain UTXOs light clients need to leverage a $T$-traced verification to realize SPV. 

In NC, a light client can use SPV to verify three things: a received transaction is included in the corresponding block; all blocks in the chain have valid proof-of-Work; the block is inserted more than $k$ blocks deeper in the chain, counting from the last block to the corresponding block. The value of $k$ can be set form one to six, depending on the transaction value. However, in Eunomia, due to the potentially existed invalidated cross-chain UTXOs, light clients have one more procedure and need to check the validity of the input UTXOs, i.e., checking whether these input UTXOs interleave with some invalidated cross-chain UTXOs. The key idea is utilizing the $T$-confirmation window. The light client needs to backtrace the input UTXOs in $T$ blocks deeper and ensures miners have confirmed their validity over the $T$-confirmation. 

For a concrete example in Fig.~\ref{fig:utxo}. Bob is a light client who receives a transaction $3$ from Alice. He first needs to carry the procedures of SPV in NC to make sure the received transaction is correct except the input UTXOs. Then Bob will verify the input UTXO of transaction $3$, which is the corresponding output UTXO of transaction $2$ in Block $B$. However, block $B$ is unconfirmed, then Bob will continue to verify its input UTXO and eventually find Block $A$ is confirmed. 

\begin{figure}[ht]
\setlength{\abovecaptionskip}{2pt}
\setlength{\belowcaptionskip}{5pt}
\centering
\includegraphics[width=0.9\linewidth]{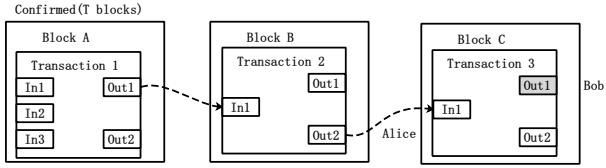}
\caption{A Simple Case for chained confirmation for SPV in Eunomia.}
\vspace{-2mm}
\label{fig:utxo}
\end{figure}



\section{Global Confirmed Block Ordering} \label{subsec:order}
\subsection{Directed graph.} In Eunomia, each node maintains a local view of the directed graph $\mathcal{D}$, depending on its receiving blocks. Each vertex is a block, and each directed edge presents a hash reference of the previous block. Upon receiving a new block $\mathcal{B}$, the node will update its local view as follows: $1)$ Only the block with a valid proof of work will be processed and stored; $2)$ If the node has the referred parent block and synchronized block (introduced later) in local view, it will compute and verify the block's clock, and then append the block to chain $\mathcal{C}_{ch(\mathcal{B})}$. Otherwise, the block will be cached until parent block or synchronized block is received. \\ 

\subsection{Virtual logical clock.} Each chain has a ``virtual'' clock, which is a counter and increments for each mined block in the chain. Here the virtual logical clock is called because chains' logical clock is not contained in the block field and is computed based on nodes' local view of directed graph $\mathcal{D}$. To synchronize clocks between chains, participants include a hash reference of the synchronized block, \emph{i.e.}, the block with the largest virtual logical clock in $m$ main chains, in the pre-mined block's header. As blocks' clock is monotonically increasing in the same chain, the synchronized block is one of the $m$ last blocks. It is easy to see that the clock synchronizing process is maintained by all participants, \emph{i.e.}, in a distributed way. 

When a new block is going to be appended to the nodes' local view, nodes will invoke Algorithm \ref{Algo1} to compute the block's virtual clocks. 
The algorithm initiates Genesis block $\mathcal{G}_i$'s logical clock as zero. Except for Genesis blocks, the algorithm first obtains block's parent block's clock and synchronized block's clock (See line $6-10$). In reality, nodes can cache the latest blocks' clocks in local memory and avoid calling the function repeatedly. The line $10-14$ ensures participants to include the block with the highest clock as the synchronized block. Otherwise, their blocks have a chance of being timestamped with invalid clock and then are refused by other nodes, which provides the incentive compatibility for obeying the protocol (addressed later). The virtual clock algorithm is designed as simply as possible, aiming to provide easily proved end-to-end security. \\

\begin{algorithm}[!ht]
\caption{logicalClockCompute or LCC in short} \label{Algo1}
\begin{algorithmic}[1]
    \State \emph{Input:} Directed graph $\mathcal{D}$, block $\mathcal{B}$ 
    \State \emph{Output:} Virtual logical clock $v$ of $\mathcal{B}$ $\in \mathbb{R}$
    \If{$\mathcal{B} \in \left ( \mathcal{G}_i \right )_{i=0}^{n}$}
		\State   $v \gets 0$ 
	\Else \Comment{Synchronize chains' clock}
	    \State  $\mathcal{B}_j$ $\gets$ AccessSynchronizedBlocks($\mathcal{D}$, $\mathcal{B}$)
        \State  {$v_j$ = logicalClockCompute($\mathcal{D}$, $\mathcal{B}_j$)} 
        \State  $\mathcal{B}_i$ $\gets$ AccessParentsBlocks($\mathcal{D}$, $\mathcal{B}$)
        \State  {$v_i$ = logicalClockCompute($\mathcal{D}$, $\mathcal{B}_i$)} 
        \If{$v_i > v_j$} \Comment{Incentive compatibility for ordering}
            \State $v \gets -1$
        \Else
		    \State   $v \gets v_j + 1$ 
        \EndIf
    \EndIf 
    \State   \Return  $v$       
\end{algorithmic}
\end{algorithm}

\subsection{Globally confirmed block sequence.}
Each nodes can output a global order on blocks by further using a global logical timestamp $(v,i)$, where $v$ is blocks' virtual logical clock and $i$ represents the blocks' belonging chain index. Specifically, blocks across $m$ main chains can be globally ordered by their clock $v$ with a tie breaking of $i$. 

However, to ensure the consistency property, nodes need further output a global order of confirmed blocks. That implicitly imply that the blocks of each chain participating in global order need to be confirmed first, which we refer to as per-chain confirmed blocks. Specifically, per-chain confirmed blocks are the blocks in each nodes' local chains except the last $T$ blocks according to Theorem \ref{lamma:1}. For example, a block in Bitcoin usually has to wait for six new blocks added to be confirmed. Thus, nodes do not consider the last $T$ unconfirmed blocks of each chain. Let $\mathcal{B}_i$ denote the last per-chain confirmed block in chain $C_i$ with logical clock $v_i$, and let $synchronized\_bar \leftarrow \min_{i=0}^{m}\left \{v_i \right \}$. Then nodes output all the per-chain confirmed blocks whose logical clock is less than $synchronized\_bar$ in the global order, \emph{i.e.}, the globally confirmed block sequence.

For a concrete example, in Fig.~\ref{fig:order} we use a round-by-round model to illustrate the block generation and ordering process. In this example, a block reaches all nodes by the end of the round and becomes per-chain confirmed after $T = 2$ confirmations. It is easy to see that block $A4$, $B7$ and $C3$ are the last per-chain confirmed blocks, and the $synchronized\_bar$ equals $4$. Thus according to the ordering rule, each honest node outputs the total block ordering: $A1, C1, A3, B1, B3, A4, B5$. Although we present the example using a synchronous model, Eunomia's security will be proved in an asynchronous model introduced later. \\

\begin{figure}[!ht]
\setlength{\abovecaptionskip}{2pt}
\setlength{\belowcaptionskip}{5pt}
\centering
\includegraphics[width=0.9\linewidth]{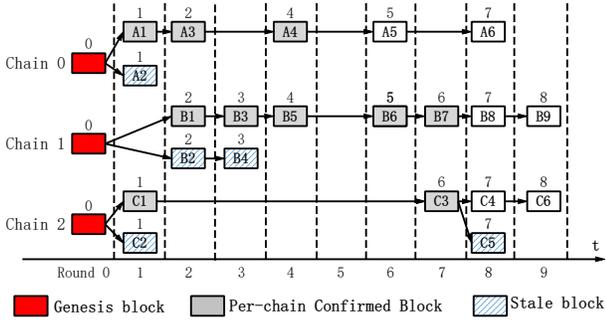}
\caption{An illustration of globally confirmed blocks}
\vspace{-2mm}
\label{fig:order}
\end{figure} 

\subsection{Incentive Design}
Eunomia adopts a similar incentive mechanism as Bitcoin \cite{nakamoto2012bitcoin}, in which participants can receive a block reward, \emph{i.e.}, some amount of self-issued tokens, for every mined block included in the globally confirmed blocks. In addition, participants can also get the transaction fees for all the confirmed transactions that are contained in the block. As the incentives in Bitcoin has been well studied \cite{eyal2014majority, sapirshtein2016optimal, gervais2016security}, the following focuses on illustrating the incentives in clock synchronization. 

In Eunomia, participants concurrently mine on $m$ last blocks of $m$ main chains and choose one block with the largest virtual clock as the synchronized block. Once a block is mined, it will synchronize its belonging chain's clock with the largest clock, resulting in an increase of $synchronized\_bar$ and global confirmation of blocks. However, some participants may deviate from the protocol and launch ordering attack, \emph{i.e.}, including arbitrary block as the synchronized block. The following analysis shows that the ordering attack will not have an impact on the security of Eunomia, and meanwhile, the attacker will suffer from economic loss. 

First, if the synchronized block either does not exist or does not pass the validity check, the referring block will not be processed by Algorithm \ref{Algo1}. Second, if the existed synchronized block is valid, but is not the block with the largest clock in the participants' view, there are two possible subcases: $1)$ Participants mine on $m$ last blocks, but use a block with a smaller logical clock as synchronized block; $2)$ Participants do not mine on the last blocks in some of the $m$ main chains, but mine on ancestor blocks with a smaller clock. Participants in the first subcase have the chance of making the block invalid for not passing the clock validity check (See line $10$-$14$ in Algorithm \ref{Algo1}.), and in the second have the chance that the mined blocks do not belong to the longest chain. It is easy to see both of these cases will make deviating participants suffer from economic loss, which provides incentive compatibility against the ordering attack. By contrast, OHIE in \cite{yu2018ohie} lacks effective methods to deal with the ordering attack.
What's more, once an honest block included in the chain will synchronize the chain with the latest clock immediately. 

\section{Security Analysis} \label{sec:analysis}
\noindent The security analysis of Eunomia is based on several prior works of NC in \cite{garay2015, garay2017, Pass2017}. Specifically, the proof composes two steps. First, each chain instance in Eunomia is proved to has the same properties as the single chain in NC. Second, extending each chain's properties further proves \emph{consistency} and \emph{liveness} of Eunomia. 

\subsection{Starting Point from NC.} 
We first give a well-proved corollary from \cite{Pass2017}, which is the base for our later analysis:

\begin{theorem} \label{theo:1}
(Corollary $3$ in \cite{Pass2017}.) Consider any given constant $\rho < \frac{1}{2}$. Then there exists some positive
constant $c$ such that for all $p < \frac{1}{c \Delta n}$, the $\Pi_{nak}$ satisfies all the following properties in $(\rho, \Delta, n, p)$-environment 
with probability that drops exponentially in $T$:
\begin{itemize}
    \item (Consistency) Let $S_1$ ($S_2$) be the sequence of blocks on the chain on any Participant $P_1$ ($P_2$) at any time $t_1$ ($t_2$), excluding the last $T$ blocks on the chain. Then either $S_1$ is a prefix of $S_2$ or $S_2$ is a prefix of $S_1$.

    \item (Chain-growth) The length of the main chain of any honest participants increases by at least $T$ blocks every $\frac{2T}{pn}$ round.
    
    \item (Chain-quality) For any $T$ consecutive blocks in any main chain of any honest participants, $(1 - \frac{\rho}{1 - \rho}) T$ blocks were contributed by honest participants.
\end{itemize}
\end{theorem}


\begin{lemma} \label{lamma:1}
Consider any give $i$ where $0 \leq i \leq m-1$, the chain $\mathcal{C}_i$ of protocol $\Pi_{enu}$ in $(m, \rho, \Delta, n, mp)$-environment has the same properties with the single chain of $\Pi_{nak}$ in $(\rho, \Delta, n, p)$-environment.
\end{lemma}

\begin{proof} Without loss of generality, we consider chain $C_0$ of protocol $\Pi_{enu}$ and shows it runs in the same way except how the blocks are generated for the chain. In $\Pi_{enu}$, each participant has a probability $mp$ to create a new block with $\log(\frac{1}{mp})$ leading zeros, and a probability $\frac{1}{m}$ to has a returned value with $\log(m)$ tailing zeroes in each query to the random oracle. As these two events are independent, miners have a probability $p$ to create a new block for chain $C_0$, which has $\log(\frac{1}{mp})$ leading zeros and $\log(m)$ tailing zeroes in each query to the random oracle. Thus, chain $C_0$ has the same probability with the single chain of $\Pi_{nak}$ to have a new block. 

In $\Pi_{enu}$, the parents' hash value is included in the transaction metadata Merkle tree, and Merkle tree root is included in the block header. Same with $\Pi_{nak}$, the parent block can never be reverted once the block is generated. That means the system knows a single parent block. Thus, blocks in chain $C_0$ are linked just as the single chain of $\Pi_{nak}$. 
\end{proof}

\subsection{Main Theorem Statements}
Eunomia is proved to satisfy the following properties:
\begin{theorem} \label{theo:2}
Consider any given constant $\rho < \frac{1}{2}$. Then there exists some positive constant $c$ and positive integer $m$ such that for all $p < \frac{1}{c \Delta n}$, the $\Pi_{enu}$ satisfies all the following properties in $(m, \rho, \Delta, n, mp)$-environment with probability that drops exponentially in $T$:
\begin{itemize}

    \item (Consistency) For any pair of honest participants $P_1$, $P_2$ outputting a globally confirmed block sequence $L_1$, $L_2$ at time $t_1 \leq t_2$, the protocol holds that $L_1$ is a prefix of $L_2$.
    
    \item (L-growth) For any integer $\gamma \geq 1$, the length of the globally confirmed block sequence $L$ of any honest participants increases by at least $m \cdot \gamma \cdot T$ blocks every $(\gamma + 2) \frac{2T}{pn}$ round.
    
    \item (L-quality) For any $m \cdot \gamma \cdot T$ consecutive blocks in the globally confirmed block sequence $L$ of any honest participants, $m \cdot \gamma \cdot (1 - \frac{\rho}{1 - \rho}) T$ blocks were contributed by honest participants.
     
\end{itemize}
\end{theorem}

Theorem \ref{theo:1} shows that the globally confirmed block sequence $L$ satisfies the consistency, quality, and growth properties. Specifically, the $L$-growth and $L$-quality states the globally confirmed block sequence $L$ will include honest miners' blocks at a certain rate. Furthermore, clients' transactions are guaranteed to be eventually processed. In other words, consistency corresponds to the safety of Eunomia, while the combination of quality and growth captures the liveness of Eunomia.


\begin{lemma} \label{lemma:1}
If the three properties in Theorem 2 hold for each of the $m$ chains of protocol $\Pi_{enu}$ in $(m, \rho, \Delta, n, mp)$-environment, then the protocol $\Pi_{enu}$ in $(m, \rho, \Delta, n, mp)$-environment satisfies the consistency property in Theorem \ref{theo:1} with probability that drops exponentially in $T$.
\end{lemma}

\begin{proof}
Let the view of participant $P_1$ at time $t_1$ be $view_1$, and the view of participant $P_2$ at time $t_2$ be $view_2$. Let $s_1$ ($s_2$) be the $synchronized\_bar$ in $view_1$ ($view_2$, respectively). Let $L_1$ ($L_2$) be the globally confirmed block sequence for $view_1$ ($view_2$, respectively). 
Let $E_1(i)$ ($0 \leq i \leq m-1$) be the per-chain confirmed block sequence of chain $i$ in $view_1$, and $G_1(i)$ be the prefix of $E_1(i)$ such that $G_1(i)$ contains all blocks in $E_1(i)$ whose logical clock is smaller than $s_1$. In the same way, we have $E_2(i)$ and $G_2(i)$. Before proving $L_1$ is a prefix of $L_2$, we first give the following claim. 

For all $i$ where $0 \leq i \leq m-1$, let $G_1(i)$ is a prefix of $G_2(i)$. According to the \textbf{consistency} property in Theorem \ref{theo:1}, $E_1(i)$ is a prefix of $E_2(i)$. Let block $B_1(i)$ ($B_2(i)$) be the last block in $E_1(i)$ ($E_2(i)$, respectively). In addition, let $v_1(i)$ ($v_2(i)$) denote the virtual logical clock of $B_1(i)$ ($B_1(i)$, respectively). As blocks' clock monotonically increase in the same chain, $v_1(i) \leq v_2(i)$. By ordering rule, we have $s_1 \gets \min_{i = 0}^{m-1}\left \{ v_1(i) \right \}$ and $s_2 \gets \min_{i=0}^{m-1}\left \{ v_2(i) \right \}$. Hence, we have $s_1 \leq s_2$. It is easy to see that $G_1(i)$ is a prefix of $G_2(i)$ with $E_1(i)$ is a prefix of $E_2(i)$ and $s_1 \leq s_2$. In addition, any block in $G_2(i) \setminus G_1(i)$ has clock with no less than $s_1$.

The globally confirmed block sequence $L_1$ in $view_1$ is constructed by that all blocks in $G_1(0)$ through $G_1(m-1)$ are ordered by their logical clock with a tie breaking of chain index. As $G_1(i)$ is a prefix of $G_2(i)$ for all $i$ $(0 \leq i \leq m-1)$, thus all blocks in $L_1$ belong to $L_2$, and are also ordered in the same way. For all block $B \in \bigcup_{i=0}^{m-1}\left \{ G_2(i) \setminus G_1(i) \right \}$, by the previous claim, $B$'s logical clock must be no smaller than $s_1$. Thus these blocks must be ordered after all the blocks in $ \bigcup_{i=0}^{m-1}\left \{  G_1(i) \right \}$ (whose logical clock must be smaller than $s_1$). Thus, It is clear that $L_1$ is the prefix of $L_2$.
\end{proof}

\begin{lemma}\label{lemma:2}
If the three properties in Theorem 2 hold for each of the $m$ chains of protocol $\Pi_{enu}$ in $(m, \rho, \Delta, n, mp)$-environment, then a per-chain confirmed block in any chain will become globally confirmed within at most $\frac{4T}{pn}$ rounds with probability that drops exponentially in $T$.
\end{lemma}

\begin{proof}
Consider any given honest participant $P$ at given time $t_0$ has the view $view$. Let $E_i(t_0)$ ($0 \leq  i \leq m- 1$) be the per-chain confirmed block sequence of chain $i$ in $view$ at time $t_0$. With loss of generality, we assume $t_0$ is after the first $\frac{2T}{pn}$ rounds of the execution, which guarantee there is at least one per-chain confirmed block in $E_i(t_0)$. Let block $B_i(t_0)$ denote the last block in $E_i(t_0)$ at time $t_0$. 

Let $x$ be the largest virtual clock of blocks $\left \{ B_i(t_0) \right \}_{i=0}^{m-1}$. By time $t_1 = t_0 + \frac{2T}{pn}$, each chain of participant $P$'s $view$ must have at least $T$ blocks attached. By the quality property in Theorem \ref{theo:1}, in the $T$ blocks, there must have one block of honest participants which synchronizes its belonging chain to the value $x$. Then from time $t_1$ to time $t_2 = t_1 + \frac{2T}{pn}$, the $T$ blocks generated from $t_0$ to $t_1$ will become partially confirmed. At time $t_2$, the $synchronized\_bar$ must be no less than $x$. Then all per-chain confirmed blocks at time $t_0$ will become globally confirmed. 
\end{proof}

\begin{lemma} \label{lemma:3}
If the three properties in Theorem 2 hold for each of the $m$ chains of protocol $\Pi_{enu}$ in $(m, \rho, \Delta, n, mp)$-environment, then the protocol $\Pi_{enu}$ in $(m, \rho, \Delta, n, mp)$-environment satisfies the L-quality and L-growth properties in Theorem \ref{theo:1} with probability that drops exponentially in $T$.
\end{lemma}

\begin{proof}
Consider any given honest participant $P$ at given time $t_0$ has the view $view$. Let $E_i(t_0)$ be the per-chain confirmed block sequence of chain $i$ in $view$ at time $t_0$. For ant integer $\gamma \geq 1$, each chain of the protocol $\Pi_{enu}$ in $(m, \rho, \Delta, n, mp)$-environment has at least $\gamma \cdot T$ blocks becoming newly partially-confirmed from time $t_0$ to time $t_1 = t_0 +  \gamma \frac{2T}{pn}$. By Lemma \ref{lemma:3}, we know after $\frac{4T}{pn}$ rounds, any of the per-chain confirmed blocks will become global confirmed. The proof is done. 
\end{proof}

\section{Discussion} \label{sec:discuss}
\subsection{Block Propagation and Stale Blocks} \label{subsec:propa}
\noindent Recall Algorithm \ref{Algo1}, 
the synchronized block must be cached in participants' local view before computing the referring block's clock, which leads to two new issues. 
First, the referred synchronized block may not be included in the main chain, which we refer to as a stale block. In Bitcoin, stale blocks are eventually abandoned by nodes. For a new node joining the network, It is hard for the node to invoke Algorithm \ref{Algo1}.   
Thus, nodes can store synchronized blocks' header and the Merkle tree proof of its included hash references, which is similar to uncle blocks in Ethereum\cite{buterin2014next}. 

Second, if the synchronized block belongs to a different chain with the referring block, that means nodes' verification and acceptance of this new arriving block will be dependent on an existing cross-chain block. In other words, each chain's block generation cannot be viewed as independent with each other anymore, which deviates from the design goal and makes the analysis complicated. However, It is not a tough issue. First, most of the nodes will receive the synchronized block before the arriving of the next block mined on it. Second, in Eunomia nodes will relay the synchronized block with newly mined blocks if their neighbors do not have the synchronized block. Thus, each chain's block generation process remains independent. 


\section{The Related Work} \label{sec:related}
\noindent In general, there are three main approaches in the literature to improve the scalability of PoW-based permissionless blockchains. The first is based on the extension of NC. The second is to leverage a more generalized graph (\emph{e.g.}, directed acyclic graph (DAG) and parallel chains) instead of the chain structure. The third is to adopt the sharding protocol to scale out. \\

\noindent \textbf{Extended NC.} The Bitcoin community in \cite{BitcoinUnlimited} proposed Bitcoin Unlimited (BU), which breaks NC's fixed block size limit and allows miners to decide the limit value collectively. However, Zhang and Preneel in \cite{Zhang2017} showed that the absence of block validity consensus (BVC) could lead to new attack vectors.
In \cite{bitcoinng}, Eyal \etal proposed Bitcoin-NG, a protocol which decouples the transaction processing from the blockchain maintenance, and significantly promotes the throughput. In Bitcoin-NG, the key blocks' owners are entitled to create many micro-blocks consisting of transactions in one epoch. However, as only a single leader is responsible for generating all micro-blocks, It is easy for a leader to launch the censorship attack, and meanwhile to suffer from DoS attacks due to the revealed IP address. In Eunomia, there are multiple concurrent leaders in charge of transaction processing, and leaders' IP addressee cannot be foreseen by the adversary until a block is mined. 

To avoid the single leader's dilemma, some hybrid consensus protocols \cite{omniledger, ByzCoin, Luu2016, rapidZamani2018, pass20171} are proposed, which combines NC with the classical byzantine agreement (BA). Specifically, NC is used to establish committees in a permissionless way, and BA is adopted to reach an agreement of transactions. The hybrid consensus protocols can provide better throughput and instant finality. However, these blockchain protocols usually assume the adversaries control less than $1/3$ of the votes or computation power. Additionally, with the increase of the committee members, the $O(n^2)$ communication complexity for reaching an agreement will become intolerable \cite{pbft1999}, and the committee reconfiguration will consume a lot of time \cite{Luu2016, omniledger,rapidZamani2018}. \\


\noindent \textbf{Generalized graph structure.} The second class of protocols is replacing the NC'S underlying chain structure with some generalized graph. The protocols such as Ghost \cite{Sompolinsky2015ghost}, Phantom \cite{sompolinsky2018phantom}, Spectre\cite{Sompolinsky2016SPECTREAF}, and Conflux \cite{li2018scaling} construct the blocks in a directed acyclic graph (DAG) by allowing a single block to reference multiple previous blocks. Particularly, Ghost adopts the heaviest subtree rule instead of LCR to choose one main chain. Phantom, Spectre and Conflux use the DAG to define the global order of all blocks. Although these protocols can leverage the forked blocks and promote the throughput, the complicated graph makes it hard to provide formal proof of their safety, liveness and throughout. 

Chainweb \cite{martino2018chainweb, Quaintance2018chainweb} is the among the first to maintain a set of parallel chains, in which each chain cross-references all other chains according to the base graph. However, Chainweb \cite{martino2018chainweb, Quaintance2018chainweb} 
does not provide a formal analysis for its claimed safety and better throughput performance. Based on the work, Kiffer \etal \cite{Kiffer2018} analyzed a special case of Chainweb's configuration, in which each chain cross-references all other chains, and showed that Chainweb has the same throughput with the single chain of NC under the same security guarantee. Attracted by the parallel chain's appealing characteristics, there are three recent concurrent work \cite{bagaria2018deconstructing, fitzi2018parallel, yu2018ohie} to propose better protocols. In \cite{fitzi2018parallel, bagaria2018deconstructing}, one of $m$ chains is designed to provide ``clock'' to synchronize all the blocks. However, the single special chain makes the whole system vulnerable to attacks or is the bottleneck for performance. In \cite{yu2018ohie}, Yu \etal utilized additional attachment blocks, which is inefficient and incompatible. Particularly, the protocols  \cite{bagaria2018deconstructing, yu2018ohie} are analyzed in a synchronous model. Sec.~\ref{sec:intro} has discussed these three works in detail. \\

\noindent \textbf{Sharding-based protocol.} In \cite{monoxide}, Jiaping Wang and Hao Wang propose a protocol called Monoxide, which composes multiple independent single chain consensus systems, called zones. They also proposed eventual atomicity to ensure transaction atomicity across zones and Chu-ko-nu mining to ensure the effective mining power in each zone. Monoxide is shown to provide $1,000x$ throughput and $2,000x$ capacity over Bitcoin and Ethereum. Except for this work, there are some committee-based sharding protocols \cite{rapidZamani2018,omniledger,Miller2016}. Each shard is assigned a subset of the nodes, and nodes run the classical byzantine agreement (BA) to reach an agreement. However, these protocols only can tolerant up to $1/3$ adversaries. What's more, all sharding-based protocols have additional overhead and latency for cross-shard transactions. 






\section{Conclusion} \label{sec:conclusion}
In this technical report, we design Eunomia, a permissionless parallel chain protocol that uses the virtual logical clock to realize the global ordering of blocks. In addition, Eunomia adopts a fin grained UTXO sharding, which does well solve the conflicting transaction issue and is shown to be SPV-friendly. Eunomia is proved to provide the end-to-end safety and liveness guarantees with less than $1/2$ adversarial computation power.

\bibliographystyle{IEEEtran}
\bibliography{IEEEabrv,reference}

\end{document}